\documentclass{article}

\usepackage[utf8]{inputenc}
\usepackage{amsmath}
\usepackage{amsthm}
\usepackage{graphicx}
\usepackage{hyperref}
\usepackage{xcolor}
\usepackage{float}
\usepackage{algorithmicx}
\usepackage{algpseudocode}
\usepackage{algorithm}
\usepackage{cite}
\usepackage{xspace}
\usepackage{tikz}
\usetikzlibrary{arrows}
\usetikzlibrary{graphs}

\newtheorem{lemma}{Lemma}
\newtheorem{theorem}{Theorem}
\newtheorem{property}{Property}
\newtheorem{remark}{Observation}
\newtheorem{corollary}{Corollary}
\theoremstyle{definition}
\newtheorem{definition}{Definition}

\newtheorem{innercustomthm}{Lemma}
\newenvironment{customlemma}[1]
  {\renewcommand\theinnercustomthm{#1}\innercustomthm}
  {\endinnercustomthm}

\newcommand{\HCtree}{HC-tree}
\newcommand{\DCcost}{\mathsf{DC\text{-}cost}}

\title{Hierarchical Clusterings of Unweighted Graphs}

\author{Svein H{\o}gemo, Cristophe Paul and Jan Arne Telle}

\begin{document}
\maketitle

\begin{abstract}
We study the complexity of finding an optimal hierarchical clustering of an unweighted similarity graph under the recently introduced Dasgupta objective function. 
We introduce a proof technique, called the normalization procedure, that takes any such clustering of a graph $G$ and iteratively improves it until a desired target clustering of $G$ is reached.
We use this technique to show both a negative and a positive complexity result. Firstly, we show that in general the problem is NP-complete. 
Secondly, we consider min-well-behaved graphs, which are graphs $H$ having the property that for any $k$ the graph $H^{(k)}$ being the join of $k$ copies of $H$ has an optimal hierarchical clustering that splits each copy of $H$ in the same optimal way. To optimally cluster such a graph $H^{(k)}$ we thus only need to optimally cluster the smaller graph $H$.
Co-bipartite graphs are min-well-behaved, but otherwise they seem to be scarce. We use the normalization procedure to show that also the cycle on 6 vertices is min-well-behaved.
\end{abstract}

\section{Introduction}


Clustering is an unsupervised machine learning technique and one of the most important problems in data-mining~\cite{Har75,CEF+06,KT09,HTF09}. Given a data set and a pairwise similarity measure, the task is to partition the data set into clusters so that similar data points belong to the same cluster.  
In a \emph{hierarchical clustering} the data set is recursively partitioned into smaller clusters, by means of a rooted binary tree whose leaves are in one-to-one correspondence with the data points. Hierarchical clustering emerged as a central task in the study of phylogenetic trees~\cite{SS63,Bun71}. Such a clustering is very general, capturing clustering structure at all levels of granularity, with a clustering into two parts given by the root of the tree, and finer clusterings given by lower levels of the tree. Algorithms for hierarchical clustering have been widely used for many years, but it was only recently that an objective function to measure their quality was formalized. In a STOC 2017 paper 
~\cite{Das16} Dasgupta introduced a natural objective function measuring the global cost of a hierarchical clustering. From now on, this function will be called the Dasgupta Clustering function - DC function. Several follow-ups to Dasgupta's work have appeared, we mention only a couple: in~\cite{CC17}, the authors improve the ratio of the approximation algorithm proposed by Dasgutpa;  in~\cite{CKM+19}, the authors revisit the DC function and propose some axioms that a "good" cost function should satisfy.

In this paper we investigate the complexity of finding the DC-optimal hierarchical clustering for {\em unweighted} similarity graphs. Thus, we assume that any pair of data points has been marked as either 'similar' or 'non-similar' and represent this information as an undirected, unweighted graph $G$ whose vertex set $V(G)$ is the set of data points and adjacencies represent similarity. We ask for an \HCtree{} (a Hierarchical Clustering tree), a rooted binary tree $T$ with leaves in one-to-one correspondence with $V(G)$, such that the DC-cost of $T$ - i.e. the sum over all edges $uv$ of $G$, of the number of leaves of the subtree rooted at the least common ancestor of $u$ and $v$ - is minimized. 
Dasgupta \cite{Das16} showed that the edge-weighted version of this problem, with weights representing degree of similarity, is NP-complete. In this paper we focus on unweighted graphs, the hardness of which was left open by Dasgupta~\cite{Das19}. 
Unweighted graphs naturally appear in this context, for example in the correlation clustering problem~\cite{BBC04}. It is also a common approach 
to transform a similarity matrix into a similarity graph by fixing a threshold value that determines whether two objects are similar or not (see~\cite{Har75} for example). We focus on dense similarity graphs.
Such graphs typically appear when there is a fixed threshold for similarity that is set to be very low, for example the existence of email correspondence within a single (small) organization, or existence of non-zero trade relations between countries. 
We show that the problem remains NP-complete, already for dense graphs. More precisely, by a reduction building on the one used in \cite{Das16}, we establish the NP-hardness for unweighted $n$-vertex graphs where every vertex has at least $n-6$ neighbours.

Note that all pairs of vertices will be split into distinct clusters at some point in the HC-tree, namely at their least common ancestor. Minimizing the DC-cost encourages pairs of adjacent vertices (similar data points) to be split lower in the tree than non-adjacent vertex pairs (non-similar data points). For example, if $G$ is the complement of a bipartite graph on color classes $A, B$ then any HC-tree $T$ that splits $A$ and $B$ at the root is optimal, which follows easily from observations in \cite{Das16} since $G[A]$ and $G[B]$ are complete graphs. 
Dasgupta showed that minimizing the DC-cost of $G$ is equivalent to maximizing the DC-cost of the complement of $G$. Thus the previous result can be restated to say that for a bipartite graph any HC-tree splitting the two color classes at the root will have max DC-cost, rendering the result trivial as all edges are now split at the root.
In the current paper we will usually take this viewpoint, thus considering {\em unweighted sparse} graphs and looking for an HC-tree {\em maximizing} the DC-cost, typically splitting pairs of adjacent vertices, now denoting non-similarity, at higher levels of the tree.

As noted, bipartite graphs are then trivial, but what other graphs can be handled efficiently? What about $G$ being a collection of disjoint copies of the same bipartite graph? Maximizing DC-cost is still trivial, in fact $G$ is again bipartite, so at the root we can simply split each copy in the same optimal way. Let us define a more complex property generalizing this behavior. Consider a graph $H$ of max DC-cost $W$ achievable by some HC-tree $T$ and let the graph $H^{(k)}$ consist of $k$ disjoint copies of $H$. If we use $T$ to simultaneously cluster each of the $k$ copies of $H$ then each leaf of $T$ will contain $k$ copies of the same vertex. These vertices induce a stable set so we can further cluster them in an arbitrary way to get an HC-tree $T^{(k)}$. Note that this tree will have DC-cost $k^2W$ since each edge of $H$ has $k$ copies in $H^{(k)}$, and the subtree of $T^{(k)}$ that splits an edge contains a multiplicative factor $k$ more vertices than the similar subtree of $T$. We call such $H$ max-well-behaved if for any $k$ the max DC-cost of $H^{(k)}$ is no higher than $k^2t$, and the complement of $H$ min-well-behaved. 

\begin{figure}[t]
\centering
\begin{tikzpicture}
[innernode/.style={circle,draw=black,fill=black,inner sep=0pt,minimum size=1mm},
leaf1/.style={circle,draw=blue,fill=blue,inner sep=0pt,minimum size=1.5mm},
leaf2/.style={circle,draw=red,fill=red,inner sep=0pt,minimum size=1.5mm}]

\node [black] at (-10.5,-3.5) {$Q_{2,3}$};

\node (i1) at (-11.1,-1) [leaf1] {} ;
\node (i2) at (-11.1,-1.8) [leaf1] {} ;
\node (i3) at (-11.1,-2.6) [leaf1] {} ;
\node (c1) at (-9.9,-1.4) [leaf2] {} ;
\node (c2) at (-9.9,-2.2) [leaf2] {} ;

\draw [-] (c1) to (c2);

\foreach \x in {1,2,3}{
	\foreach \y in {1,2}{
		\draw [-] (i\x) to (c\y) ;
		\node [red,right] at (c\y.east) {$c_{\y}$};
		}
	\node [blue,left] at (i\x.east) {$s_{\x}~$};
	}

\node [black] at (-7+0.25,-3.5) {$T$};

\node (y0) at (-7+0.25,-1) [innernode] {} ;
\node (y1) at (-8+0.25,-1.5) [innernode] {} ;
\node (y2) at (-6+0.25,-1.5) [innernode] {} ;
\node (y3) at (-8.5+0.25,-2) [innernode] {} ;

\draw [-] (y0) to (y1);
\draw [-] (y0) to (y2);
\draw [-] (y1) to (y3);

\foreach \x in {1,2}{
	\node (t\x) at (-7+0.25-\x,-2.5) [leaf1] {} ;
	\node [blue,below] at (t\x.south) {$s_{\x}$};
}

\node (t3) at (-7.5+0.25,-2) [leaf1] {} ;
	\node [blue,below] at (t3.south) {$s_3$};

\draw [-] (y3) to (t1);
\draw [-] (y3) to (t2);
\draw [-] (y1) to (t3);

\foreach \x in {1,2}{
	\node (u\x) at (-7.5+0.25+\x,-2) [leaf2] {} ;
	\node [red,below] at (u\x.south) {$c_{\x}$};
}

\draw [-] (y2) to (u1);
\draw [-] (y2) to (u2);

\node [black] at (-2,-3.5) {$T'$};

\node (a0) at (-2,0) [innernode] {} ;
\foreach \x in {1,2,3,4}{
	\node (a\x) at (-\x/2-2,-\x/2) [innernode] {} ;
	\node (b\x) at (\x/2-2,-\x/2) [innernode] {} ;
	}

\foreach \x in {1,2}{
	\node (i\x) at (-4.75+\x/2,-2.5) [leaf1] {} ;
	\node [blue,below] at (i\x.south) {$s_{\x}$};
	}
\node (i3) at (-3.25,-2) [leaf1] {} ;
\node [blue,below] at (i3.south) {$s_3$};

\node (j1) at (-0.75,-2) [leaf1] {} ;
\node [blue,below] at (j1.south) {$s'_{1}$};
\foreach \x in {2,3}{
	\node (j\x) at (0.75-2+\x/2,-2.5) [leaf1] {} ;
	\node [blue,below] at (j\x.south) {$s'_{\x}$};
	}

\foreach \x in {1,2}{
	\node (c\x) at (-\x/2-2+1/4,-\x/2-1/2) [leaf2] {} ;
	\node [red,below] at (c\x.south) {$c'_{\x}$};
	\node (d\x) at (\x/2-2-1/4,-\x/2-1/2) [leaf2] {} ;
	\node [red,below] at (d\x.south) {$c_{\x}$};
	}

\foreach \x/\y in {0/1,1/2,2/3,3/4}
	\draw [-] (a\x) to (a\y);
\draw [-] (a0) to (b1) ;
\foreach \x/\y in {1/2,2/3,3/4}
	\draw [-] (b\x) to (b\y);

\foreach \x in {1,2}{
	\draw [-] (a\x) to (c\x);
	\draw [-] (b\x) to (d\x);
}

\draw [-] (i1) to (a4);
\draw [-] (i2) to (a4);
\draw [-] (i3) to (a3);
\draw [-] (j1) to (b3);
\draw [-] (j2) to (b4);
\draw [-] (j3) to (b4);

\end{tikzpicture}
\caption{The complete split graph $Q_{2,3}$
is not max-well-behaved. We have $\DCcost(Q_{2,3},T)= 6 \times 5 + 1 \times 2 = 32$ which is the maximum possible. The HC-tree $T'$ of $Q_{2,3}^{(k)}$ with $k=2$ (vertices $s_1,c_1,...$ in one copy and $s'_1,c'_1,...$ in the other copy) satisfies $\DCcost(Q_{2,3}^{(k)},T')=130$ which is larger than $\DCcost(Q_{2,3},T)\times k^2=128$, i.e. the DC-cost of the factorized \HCtree{} clustering both copies according to $T$ simultaneously. }\label{fig:max-well-behave}
\end{figure}
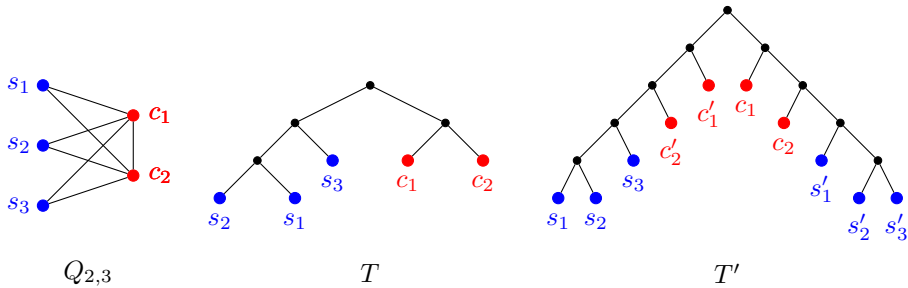

We have argued that any bipartite graph is max-well-behaved, but this is not the case for all $H$. 
For a simple example, in Figure 1 we see that complete split graphs are not max-well-behaved. 
In this paper, as a spin-off of our NP-completeness proof, we initiate the study of well-behaved graphs. We introduce a normalization procedure that makes incremental changes to a given HC-tree of some $H^{(k)}$, while observing monotonicity in the DC-cost, to arrive at a new HC-tree showing that $H$ is well-behaved. We employ this to show that the prism graph (the complement of a 6-cycle) is max-well-behaved, and thus $C_6$ min-well-behaved, establishing  the aforementioned NP-completeness along the way. 

\section{Preliminaries}
We use standard graph-theoretic notation \cite{Die05}. 
A hierarchical clustering of a similarity graph $G = (V,E)$ is a full rooted binary tree $T$, together with a bijection $\delta$ from $V$ to $L(T)$, the set of leaves of $T$. We call such a pair $(T, \delta)$ an HC-tree of $G$. For a node $t$ of $T$ we denote by $T[t]$ the subtree of $T$ rooted at $t$.
The Dasgupta cost function \cite{Das16} is this (lca means least common ancestor):
$$\DCcost(G,(T,\delta)) = \sum_{uv\in E}w(uv)\cdot|L(T[x])|:x\text{ is the lca of }\delta(u)\text{ and }\delta(v)$$
and an HC-tree of minimum DC-cost (under Dasgupta's objective function) is thus an HC-tree $(T^*,\delta^*)$ that minimizes DC-cost. 

Dasgupta shows that any HC-tree with minimum weight for graph $G$ is also an HC-tree with maximum weight for its complement $\overline{G}$.
We consider only unweighted graphs, equivalently $w(uv)=1$ for all $uv \in E$ and $0$ otherwise. 
For any node $t\in T$, we define $G_{(T,\delta)}[t]$ as the subgraph of $G$ induced by $\delta^{-1}(L(T[t]))$, the vertices of $G$ mapped to leaves in $T[t]$.
Similarly, for any two nodes $t_1,t_2\in T$ with $L[t_1]\cap L[t_2] = \emptyset$, we define $G_{(T,\delta)}[t_1,t_2]$ 
as the bipartite subgraph of $G$ consisting of all edges with one endpoint in  $\delta^{-1}(L(T[t_1]))$ and the other endpoint in $\delta^{-1}(L(T[t_2]))$.
If $(T,\delta)$ is inferred from context, we further shorten these to $G[t]$ and $G[t_1,t_2]$.
We can now simplify the Dasgupta cost function on unweighted graphs as follows:
$$\DCcost(G,(T,\delta)) = \sum_{t\in V(T)\backslash L(T)} |V(G[t])|\cdot|E(G[c_l,c_r])|:c_l,c_r\text{ children of }t$$


We start with a simple but useful fact.

\begin{property}\label{Fact:EdgeDisjSum}
Let $G,G'$ be two edge-disjoint graphs over the same vertex set $V(G)$, and $(T,\delta)$ an HC-tree of $V$. The DC-cost of the decomposition on their union $G^U = (V(G),E(G)\cup E(G'))$ is the sum of the costs on each graph:
$$\DCcost(G^U,(T,\delta)) = \DCcost(G,(T,\delta))+\DCcost(G',(T,\delta))$$
\end{property}

\begin{proof}
The cost of $(T,\delta)$ on $G^U$ is simply the sum, over every edge $e\in E(G^U)$, of the size (i.e. number of vertices) of the subgraph in which $e$ is cut. This is the same as adding together the sums over every edge in $G$ and every edge in $G'$.
\end{proof}

\begin{corollary}[\emph{\cite{Das16}, Section 4.1}]\label{Lemma:MinMaxEquiv}
An HC-tree of $G$ with minimum DC-cost is also an \HCtree{} of $\overline{G}$ with maximum DC-cost.
\end{corollary}

\begin{proof}
$\overline{G}$ is by definition edge-disjoint from $G$, therefore 
$\DCcost(G^U,(T,\delta)) = \\ \DCcost(G,(T,\delta))+\DCcost(\overline{G},(T,\delta))$
by Property \ref{Fact:EdgeDisjSum}. But the union of $G$ and $\overline{G}$ is isomorphic to $K_n$ where $n=|V(G)|$, and we know that every \HCtree{} of $K_n$ has the same cost, namely $\frac{1}{3}(n^3-n)$ (\cite{Das16}, Theorem 3). Therefore, for any \HCtree{} $(T,\delta)$,
$\DCcost(\overline{G},(T,\delta)) = \frac{1}{3}(n^3-n)-\DCcost(G,(T,\delta))$.
We conclude that a \HCtree{} of $G$ with minimum cost is a \HCtree{} of $\overline{G}$ with maximum cost, and vice versa.
\end{proof}

\section{Well-behaved Graphs}\label{Sct:MaxCost}

Minimizing DC-cost of a graph is accomplished by the exact same HC-trees that maximize DC-cost for the complement graph. However, for specific graph classes, like bipartite graphs, it can be easy to find an HC-tree maximizing the DC-cost but hard to minimize the DC-cost, or vice-versa.
Let us consider a very simple operation to construct sparse graphs. Take $G^{(k)}$, consisting of $k$ disjoint copies of some graph $G$. If we are given an HC-tree $T$ for $G$ of minimum DC-cost then any HC-tree for $G^{(k)}$ hierarchically clustering each copy of $G$ as done in $T$ will have minimum DC-cost.
However, maximizing the DC-cost for $G^{(k)}$ seems harder. Given an HC-tree $T$ of maximum DC-cost for $G$ we call any HC-tree for $G^{(k)}$ that hierarchically clusters each copy of $G$ as in $T$ a factorized \HCtree{}. Let us define this formally:

\begin{definition}[Factorized \HCtree]
Let $G$ be a graph, $(T,\delta)$ an HC-tree of $G$ of maximum DC-cost $W$, and $k$ a natural number. A \emph{factorized \HCtree} $(T,\delta)^{(k)}$ of the graph $G^{(k)}$ is made as follows: Make a copy of $(T,\delta)$ and for every node $t$, make
$$G^{(k)}_{(T,\delta)^{(k)}}[t] = \bigcup_{i=1}^k G_{(T,\delta)}[t]$$
This is not a complete \HCtree{}, since for $t\in L(T)$, $G^{(k)}[t]$ is not a single vertex, but $k$ vertices. But these $k$ vertices are all disjoint, therefore any extension of this partial \HCtree{} will have the same DC-cost $k^2W$ and be regarded as a factorized \HCtree{}. 
\end{definition}

As previously mentioned, if $G$ is bipartite then for any $k$ the factorized \HCtree{} for $G^{(k)}$ will have max DC-cost. We give this property a name.

\begin{definition}[Well-behaved graph]
Let $G$ be an unweighted graph, and $W$ the maximum DC-cost over HC-trees of $G$. We call $G$ \emph{max-well-behaved}, or just \emph{well-behaved} if, for any natural number $k$, the maximum Dasgupta cost over HC-trees of the graph $G^{(k)}$ is equal to $k^2W$. The complementary graph $\overline{G}$ is called \emph{min-well-behaved}.
\end{definition}

So any bipartite graph $G$ is well-behaved and thus computing the max DC-cost of any $G^{(k)}$ can be reduced to computing the max DC-cost of $G$, or equivalently, computing the min DC-cost of $\overline{G^{(k)}}$ (the \emph{join} of $k$ copies of $\overline{G}$) reduces to computing the min DC-cost of $\overline{G}$. We may naturally ask: Is every graph well-behaved?  On the contrary, counterexamples abound, even for very small graphs, see Figure \ref{fig:max-well-behave} for an example.

How to show that some interesting non-bipartite graph $G$ is well-behaved? We need to show that for any value of $k$ no \HCtree{} of $G^{(k)}$ has higher DC-cost than the factorized \HCtree{}. We will show this by what we call a normalization procedure on HC-trees: starting with an arbitrary \HCtree{} we incrementally, step by step, modify it into the factorized \HCtree{} and show that at no step does the cost decrease. We formalize this notion:

\begin{definition}[Safe operation]
An operation that takes an \HCtree{} of a graph $G$ as input and outputs another \HCtree{} of the same graph is called {\em safe} (for maximization) if the DC-cost of the input is no larger than the DC-cost of the output.
\end{definition}

\begin{property}\label{Fact:ExistsProc}[Normalization Procedure]
Let $G$ have max HC-tree $(T,\delta)$. If there is a procedure that for any $k$ takes as input any HC-tree of $G^{(k)}$, iteratively applies safe operations, and outputs a factorized \HCtree{} $(T,\delta)^{(k)}$ of $G^{(k)}$ then $G$ is well-behaved.
\end{property}

The prism $P$ is the graph on six vertices shown in Figure \ref{fig:prismvars}.
It is non-bipartite, and its complement is a cycle.
$P$ exhibits a high degree of symmetry (it is vertex-transitive), and thus has a limited number of non-isomorphic decompositions. The optimal \HCtree{} we will base our normalization procedure around is also shown in Figure \ref{fig:prismvars}, and has the maximum cost of 48 (note $P$ has also another optimal HC-tree{}). To be convinced that this is indeed optimal, note that in a minimum optimal \HCtree{} $(T,\delta)$ of its complement, every subgraph induced by a node in $T$ must be connected if the whole graph is connected.
We
will show in Section \ref{Sct:Proc} a normalization procedure for the prism as described in Property \ref{Fact:ExistsProc} to establish the following:

\begin{lemma} \label{Lemma:PrismMax}
The prism is max-well-behaved, and thus $C_6$ is min-well-behaved.
\end{lemma}

This result is non-trivial, and should be seen in light of e.g. the five-vertex graph in Figure \ref{fig:max-well-behave}, whose complement is a 3-cycle and two isolated vertices, that is not max-well-behaved.




\section{NP-Hardness for Unweighted Graphs}

Dasgupta shows that for edge-weighted graphs, finding an HC-tree of maximum DC-cost is NP-hard, by reduction 
from an NP-complete problem he called NAESAT*:

\begin{definition}[NAESAT*]
We are given a boolean CNF formula where every clause contains either two or three literals (called "2-clauses" and "3-clauses", respectively), and every variable appears in exactly one 3-clause, and in exactly two 2-clauses with one appearance positive and the other negative. Moreover, no 2-clause nor its copy with polarities reversed is part of any 3-clause. Is there a not-all-equal-satisfying assignment, i.e. one where every clause contains at least one true and one false literal?
\end{definition}


Dasgupta first gave a simple reduction from NAE3SAT, where every clause has exactly 3 literals but there is no restriction on how many times each variable appears in the formula, to NAESAT*. In that reduction it follows trivially that no 2-clause nor its copy with polarities reversed will be contained in a 3-clause, so we have included that property in our definition of NAESAT*.
We will assume, as Dasgupta \cite{Das19} does, that if there is a 2-clause $C$ whose literals also appear in a 2-clause $C'$, but with reversed polarity, then $C'$ is removed.


Dasgupta's reduction to hierarchical clustering takes as input a NAESAT* formula $\varphi$ on $n$ variables with $m=\frac{1}{3}n$ 3-clauses and $m'\leq n$ 2-clauses, and constructs a graph $G$ with two vertices for each variable $x$ appearing in the formula $\varphi$: one corresponding to $x$ and one to $\overline{x}$. For every 2-clause $(\tilde{x}\vee\tilde{y})$, where a variable with a tilde above, $\tilde{x}$, is shorthand for "$x$ or $\overline{x}$", he adds an edge between $\tilde{x}$ and $\tilde{y}$, and also between $\overline{\tilde{x}}$ and $\overline{\tilde{y}}$ (these $2m'$ edges are called the \emph{2-clause edges}). For every 3-clause $(\tilde{x}\vee\tilde{y}\vee\tilde{z})$, he adds a triangle between $\tilde{x}$, $\tilde{y}$ and $\tilde{z}$, and also between $\overline{\tilde{x}}$, $\overline{\tilde{y}}$ and $\overline{\tilde{z}}$ (these $6m$ edges are called the \emph{3-clause edges}). In addition, he adds one edge between $x$ and $\overline{x}$ for every variable (these $n$ edges are called the \emph{matching edges}). He shows that $\varphi$ is in NAESAT* if and only if $G$ has {\em weighted} DC-cost at least $M$ (for some fixed $M$ that we do not specify here). Let us see how this comes about. Given a not-all-equal assignment of truth values to the $n$ variables of $\varphi$, he constructs an HC-tree of $G$ 
by first splitting $V(G)$ evenly at the root into True literals and False literals and then  splitting all remaining edges at the next level.

This \HCtree{} cuts all $n$ matching edges at the top since $x$ and $\overline{x}$ have opposite truth values. Since the assignment is not-all-equal satisfying all $2m'$ 2-clause edges are cut at the top, and also $4m$ of the $6m$ 3-clause edges are cut at the top. Thus $4m+2m'+n$ are cut at the top. The remaining $2m$ 3-clause edges are all disjoint, without sharing any endpoints, and can thus be cut in one single split at the level below the root. Dasgupta in his reduction gives a high weight to the matching edges (specifically, the matching edges have weight $2nm+1$) to ensure that any HC-tree of weighted DC-cost $M$ will be a tree that cuts all matching edges at the top. Note that an HC-tree cutting all matching edges at the top will naturally define a truth assignment to the variables of the formula. We will show the same result even when all edges have unit weight; this will imply the following:

\begin{theorem}
Hierarchical clustering of unweighted graphs is NP-hard.
\end{theorem}

\begin{proof}
Let the graph $G$ constructed by the Dasgupta reduction when given $\varphi$ be unweighted. What is then the cost of the HC-tree described above on $G$, given some not-all-equal assignment of the underlying Boolean formula $\varphi$? As described above, in $G$ there are $4m+2m'+n$ edges that are cut at the top and each receive a cost of $2n$, and $2m$ edges that are split at the next level and each receive a cost of $n$. The total cost is thus $W^*=10nm+4nm'+2n^2$. We have already argued that if $\varphi$ is not-all-equal-satisfiable then DC-cost of $G$ is at least $W^*$, but now we need to argue the converse. If we restrict to HC-trees that split $V(G)$ into two equally big parts, then we see that $W^*$ is the maximum possible and it can only be reached if the resulting assignment is not-all-equal satisfying. This is because it will have to cut all matching edges at the top and furthermore there is no way to cut more than two edges of a triangle in a single split.

It remains to show that an HC-tree not splitting $V(G)$ evenly at the top will have DC-cost less than $W^*$. To this purpose, we partition the edges of $G$ into two subgraphs $G'$ and $G''$, with $G'$ being the graph containing only the $2m'$ 2-clause edges, and $G''$ containing the 3-clause edges and matching edges. We observe that the 3-clause edges comprise $2m$ disjoint triangles, and that the matching edges bind together pairs of triangles, as shown in Figure \ref{fig:prismvars}. This means that $G''$ is a collection of $m$ disjoint prisms.
The graph $G'$ is also easy to describe; every variable appears in either one or two 2-clauses. It will belong to a single 2-clause when there was a 2-clause $C$ whose literals also appeared with reversed polarity in a 2-clause $C'$ and $C'$ was removed, otherwise it will belong to two 2-clauses.
Thus $G'$ will be a collection of  disjoint components that are 
1-regular (single edges) or 
2-regular (cycles). Since $G'$ is a collection of edges and cycles it is easy to see that no HC-tree whose root is an uneven split can cut all its $2m'$ edges at the top.
From Property \ref{Fact:EdgeDisjSum} we know that
for an HC-tree $(T,\delta)$ of $G$ we have
 $\DCcost(G,(T,\delta))=\DCcost(G',(T,\delta))+ \DCcost(G'',(T,\delta))$. Thus, for an uneven \HCtree{} $(T,\delta)$ of $G$ to have cost at least $W^*$, then $\DCcost(G'',(T,\delta)')$ must be strictly higher than $W^*-4nm'$ since $G'$ would contribute less than $4nm'$. By the equality $n=3m$, we get 
$$W^*-4nm' = 10mn+2n^2 = 30m^2+18m^2 = 48m^2$$
 so that $G''$ must contribute more than $48m^2$. But our main Lemma \ref{Lemma:PrismMax} showing that the prism is well-behaved, implies that $48m^2$ is the maximum cost achievable for $G''$ being $m$ copies of the prism. It must then be the case that there is no uneven \HCtree{} of $G$ with cost at least $W^*$.

We conclude that there exists an HC-tree of $G$ with weight at least $10nm+4nm'+2n^2$ if and only if the underlying Boolean formula is not-all-equal satisfiable.
\end{proof}

\begin{figure}[ht!]
\centering
\includegraphics[width=\textwidth]{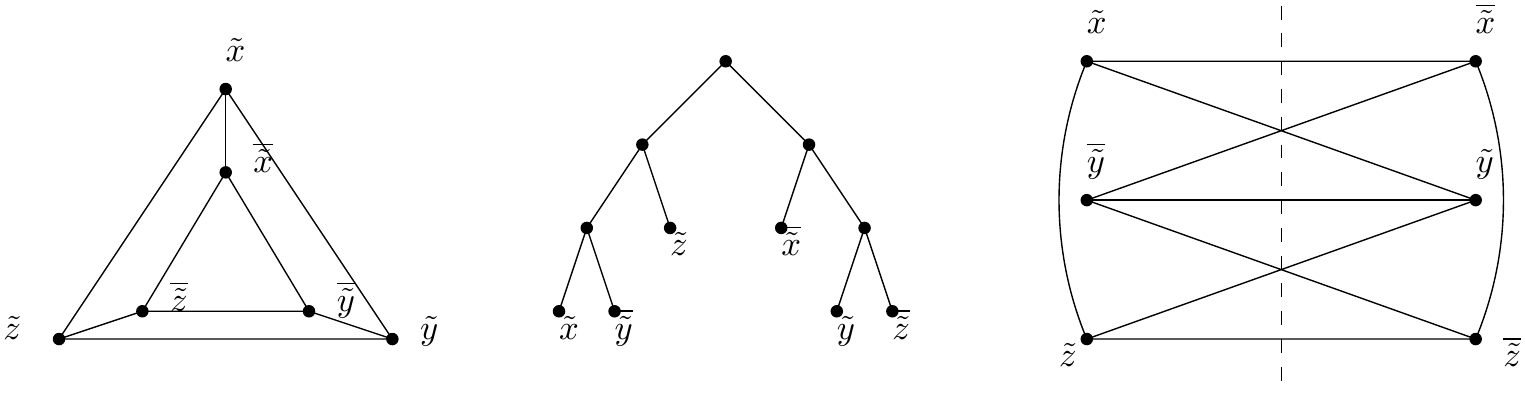}
\caption{The prism $P$, made from 3-clause edges and matching edges. By our definition of NAESAT*, every 3-clause in $\varphi$ is represented in $G$. To the middle and right, one possible HC-tree of $P$ with maximum DC-cost, and the top split of this tree.}
\label{fig:prismvars}
\end{figure}


\section{The Normalization Procedure} \label{Sct:Proc}
We give a normalization procedure for $G=P^{(k)}=P_1\cup P_2\cup\ldots\cup P_k$ consisting of $k$ disjoint copies of the prism $P$. This procedure takes as input an HC-tree for $G$, performs a series of safe operations, and outputs a factorized \HCtree{} where every prism is clustered according to the evenly balanced HC-tree $T$ in Figure \ref{fig:prismvars}.
We could have done this naively by a single Bottom-Up traversal of the tree, performing some PowerfulBalancing operation on each node $t$ of the tree. For every possible split of a subgraph of a prism at node $t$, PowerfulBalancing would have to perform a safe operation that changes this split into one that is closer to the desired end goal. 
However, the number of subgraphs of a prism, and the number of distinct splits of these subgraphs is very high, 11 and 83 respectively. Thus the naive PowerfulBalancing is not a practical option to try and prove that the prism is well-behaved. Instead, our normalization procedure will lower the number of distinct subgraphs and splits of these subgraphs that appear in a node of the tree before doing the Balancing.  In total, we employ 3 subroutines at each node $t$ of the tree: 
\begin{itemize}
    \item Cut Optimization: ensures that every sub-prism split at $t$ involves one of the 6 subgraphs given in Figure \ref{fig:optsub} and is split according to one of 8 specific splits plus 6 distinct mirror-images.
    \item Left-Heavy Distribution: ensures that no sub-prism split at $t$ has the subgraph in the right child bigger than the one in the left child,  restricting to the 8 distinct splits; Figure \ref{fig:optsplit} depicts these splits.
    \item Balancing: ensures that every sub-prism split at $t$ is split as evenly as possible
\end{itemize}
The normalization procedure will make 2 traversals of the tree: the first is a Top-Down traversal that will perform Cut Optimization on each node, the second is a Bottom-Up traversal that on each node will perform Left-Heavy Distribution followed by Balancing.

\begin{algorithm}[ht!]
\begin{algorithmic}
\Function{Normalize}{$G$:graph, $(T,\delta)$:\HCtree, $t\in V(T)$}
\If {$t\in L(T)$}
	\State \Return
\EndIf
\State $c_l,c_r\gets$ Children of $t$ in $T$
\State $\delta\gets$ Cut Optimization (cf. Section \ref{Sct:CutOpt}) on $\delta$ with regards to $G[t]$
\State\Call{Normalize}{$(T,\delta),c_l$}
\State\Call{Normalize}{$(T,\delta),c_r$}
\State $(T,\delta)\gets$ Left-Heavy (cf. Section \ref{Sct:LeftHeavy}) on $(T,\delta)$ with regards to $G[t]$
\State $(T,\delta)\gets$ Balancing Out (cf. Section \ref{Sct:Balance}) on $(T,\delta)$ with regards to $G[t]$
\EndFunction
\\
\Function{Normalization}{$G$:graph, $(T,\delta)$:\HCtree}
\State $r\gets$ Root of $T$
\State \Call{Normalize}{$G$,$(T,\delta)$,$r$}
\EndFunction
\end{algorithmic}
\caption{This pseudocode outlines in which manner the subroutines are called on the \HCtree{} $(T,\delta)$.}
\end{algorithm}

For every prism $P_i$ in $G$ and every internal node $t$ in $T$, we define $P_i[t]$ to be the subgraph of $P_i$ that lies inside the cluster at $t$: $P_i[t]=P_i\cap G[t]$. Each step of the procedure works on each of these subgraphs, striving to optimize the way these subgraphs are split.

In the next section we show that after the Cut Optimization is done on all nodes of the tree, every subgraph $P_i[t]$ is one of the six subgraphs $S_1,\ldots,S_6$ that are depicted in Figure \ref{fig:optsub}. This means that in the continuation we only have to consider splits involving these subgraphs.

We introduce some symbolic notation to easily talk about these splits. Let $t$ be an internal node in the \HCtree{}  $T$ and let $c_l$ and $c_r$ be its children. Let $P_i[t]$ be any subgraph. If we have done Cut Optimization on $(T,\delta)$, we know that $P_i[t]$, $P_i[c_l]$ and $P_i[c_r]$ are isomorphic to some $S_a$, $S_{a_l}$ and $S_{a_r}$, respectively. Then we denote the \emph{split of $P_i$ at $t$} as $S_a\rightarrow(S_{a_l},S_{a_r})$.\\

\begin{figure}[ht!]
\centering
\includegraphics[width=0.9\textwidth]{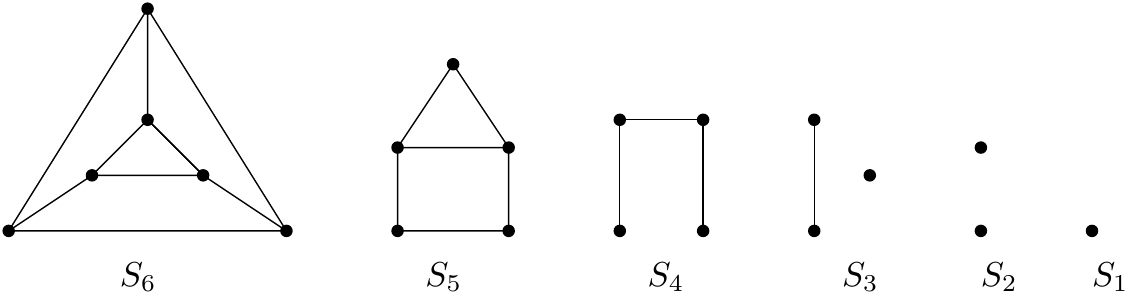}
\caption{The sub-prisms arising from optimal splits}
\label{fig:optsub}
\end{figure}

We must say a few words on what it means for a subtree of an \HCtree{} to be fully normalized, i.e. after we have performed Balancing on the root of the subtree. The end goal is clear: when we are finished, i.e. when we have performed Balancing on the root $r$ of $T$, we want every prism being split into two $S_3$'s at the root, and those $S_3$'s split into $S_2$'s and $S_1$'s at the children of the root, as seen in Figure \ref{fig:prismvars}. But when dealing with the subtree $T[t]$ for a node $t$ further down the tree, the subgraphs involved can be any $S_a$. Therefore we define "fully normalized" as every such $S_a$ in the subtree $T[t]$ being split the same way, for all $a$. The allowed splits are $S_6\rightarrow(S_3,S_3)$, $S_5\rightarrow(S_3,S_2)$, $S_4\rightarrow(S_2,S_2)$ and $S_3\rightarrow(S_2,S_1)$.

The next sections are devoted to proving that in our normalization procedure, both the top-down traversal is a safe operation, performing Cut Optimization on every node, and also the subsequent bottom-up traversal is a safe operation, performing Left-Heavy Distribution followed by Balancing on every node of the tree.






\subsection{Cut Optimization}\label{Sct:CutOpt}

Let $G=P^{(k)}$ be $k$ disjoint prisms, and let $(T,\delta)$ be any \HCtree{} of $G$. We look at some node $t\in T$. Every subgraph $P_i[t]$ is split into two subgraphs $P_i[c_l]$ and $P_i[c_r]$, with some $r$ and $s$ vertices, respectively. Not every way to split one graph into two subgraphs with given numbers of vertices is equally good. The optimal split of $P_i[t]$ into subgraphs with $r$ and $s$ vertices, is simply the split that cuts the most edges.

\begin{remark} \label{obs:optsplit}
Let $G$ and $(T,\delta)$ as above. Let $t$ be an internal node in $T$ with children $c_l,c_r$, and assume that some $P_i[t]$ is split optimally. Furthermore, let $S_1,\ldots,S_6$ be the graphs depicted in Figure \ref{fig:optsub}. Whenever $P_i[t]=S_a$ for some $a$, then $P_i[c_l]=S_{a_l}$ and $P_i[c_r]=S_{a_r}$ for some $a_l,a_r$.
\end{remark}

\begin{proof}
It is not hard to verify via simple counting that the subgraphs $S_1,\ldots,S_6$ have the minimal number of edges among the subgraphs of the prism. Since there, for any $S_a,S_b$ with $a+b\leq 6$, exists a split of $S_{a+b}$ into $S_a$ and $S_b$, this split must cut more edges than any other split of $S_{a+b}$.

Obtaining an optimal split is thus a matter of simply switching around vertices between $P_i[c_l]$ and $P_i[c_r]$. Formally, switching vertices $u$ and $v$ in $G$ with respect to $(T,\delta)$ can be seen as an operation on $\delta$, yielding a new bijection $\delta'$ with the property that $\delta(u)=\delta'(v)$, $\delta(v)=\delta'(u)$, and for every vertex $w\neq u,v$, $\delta(w)=\delta'(w)$. This operation preserves the size of every subgraph of $G$ induced by $(T,\delta)$, therefore the only edges affected are the ones that lie on $u$ or $v$. We thus conclude that every split that cuts some $S_a$ optimally, cuts it into $S_{a_l},S_{a_r}$ for some $a_l,a_r$.
\end{proof}

\begin{figure}[ht!]
\centering
\includegraphics[width=0.8\textwidth]{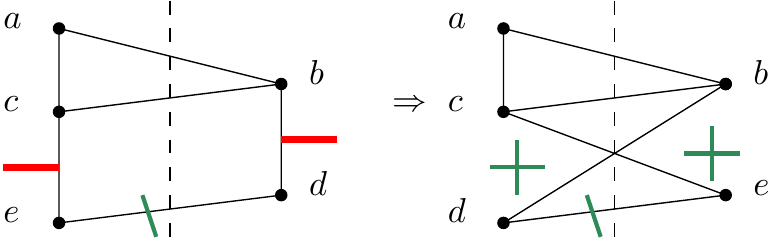}
\caption{In Cut Optimization, we obtain an optimal cut from a suboptimal one by switching two vertices, in this case $d$ and $e$. Note that $b$ and $c$ could also be used.}
\label{fig:bettersplit}
\end{figure}


\begin{lemma}\label{Lemma:safepart1}
For any node $t\in T$, Cut Optimization on $(T[t],\delta)$ is a safe operation.
\end{lemma}

\begin{proof}
From the proof of Remark \ref{obs:optsplit}, we see that for all $P_i[t]$ that is isomorphic to some $S_a$, performing Cut Optimization is a safe operation, as it never decreases the DC-cost of $(T,\delta)$. Now, note that we perform this operation on each node of $T$ in top-down fashion. At the root of $T$, $r$, we have that for every $1\leq i\leq k$, $P_i[r]=P=S_6$, so the operation is safe on $r$. At any other node $t$, we have already optimized the cuts in $u$, the parent of $t$. By Remark \ref{obs:optsplit}, we again have that for every $1\leq i\leq k$, there exists some $a$ such that $P_i[t]=S_a$. Therefore, the operation also is safe on every other node of $T$.
\end{proof}


\begin{figure}[ht!]
\centering
\includegraphics[width=0.8\textwidth]{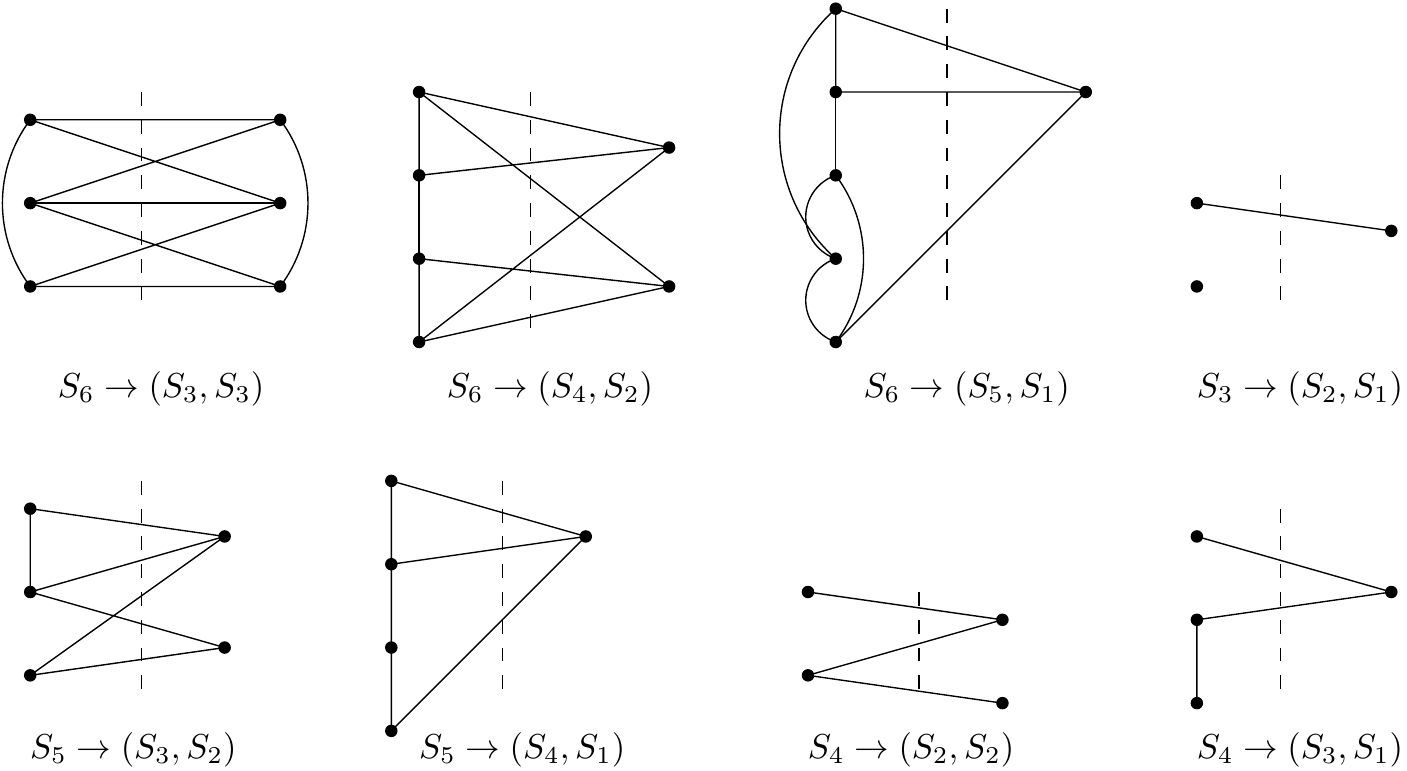}
\caption{After Cut Optimization, every split of sub-prisms that cuts at least one edge is one of the splits shown here or its mirror image. After Left Heavy the mirror images no longer appear.}
\label{fig:optsplit}
\end{figure}

\subsection{Left-Heavy Distribution}\label{Sct:LeftHeavy}

Now we show that also Left-Heavy Distribution is a safe operation on each node. This step is performed after Cut Optimization, therefore we can assume every split in the \HCtree{} is an optimal one. Furthermore, since this step is done in tandem with the Balancing step, on each node before moving up to its parent, we can assume that when performing Left-Heavy Distribution on some node $t$ in $T$ with children $c_l$ and $c_r$, then $T[c_l]$ and $T[c_r]$ are already fully normalized.

The goal of the second step, Left-Heavy Distribution, is to ensure that for every $i$, $|P_i[c_l]|\geq|P_i[c_r]|$. The intuition behind this step is clear: if we first split one component unevenly, we would expect more uncut edges in the big part than in the small part. Indeed, this is true for the subgraphs $S_1,\ldots,S_6$; $S_a$ does not have more edges than $S_{a+1}$ for any $a \in \{1,\ldots,5\}$. Splitting all components unevenly with the big part on the same side, we give more weight to these remaining edges when they are cut, further down in $T$.

We begin by dividing $G[t]$ into two pieces, $G[t]^L$ and $G[t]^R$. $G[t]^L$ is the union of all those $P_i[t]$ for which $|P_i[c_l]|\geq|P_i[c_r]|$ (the left-heavily split subgraphs), while $G[t]^R$ is the union of all those $P_i[t]$ for which $|P_i[c_l]|<|P_i[c_r]|$ (the right-heavily split subgraphs). $G[t]^L$ and $G[t]^R$ are clearly disjoint, since every connected subgraph lies wholly within one of these parts. We make a couple of observations about these two subgraphs:

\begin{remark}\label{Rmk:RightSideEdges}
Every edge in $G[c_l]$ is also in $G[t]^L$, and every edge in $G[c_r]$ except those arising from (3-3)-splits is also in $G[c_r]$.
\end{remark}

\begin{proof}
We begin looking at $G[c_l]$: As we have performed Cut Optimization on the \HCtree{}, we can assume that $P_i[c_l]$ is isomorphic to $S_{a_l}$ for some $a_l\in\{0,\ldots,6\}$ for every $i$, and equivalently every $P_i[c_r]$ is isomorphic to some $S_{a_r}$. Now, for any $P_i[t]$, if this subgraph has been put into $G[t]^R$ it is because it has been split right-heavily, i.e. $a_l<a_r$. Since $a_l+a_r$ is at most 6, is follows that $a_l$ is at most 2. But the optimal subsets of the prism that contain edges all have at least 3 vertices, therefore $P_i[t]$ cannot contain any edges.

The proof for $G[c_r]$ is roughly equivalent to the one above, but we have to factor in that there can exist some $P_i[c_r]$ in $G[t]^L$ that is isomorphic to $S_3$. If this is the case, then we know that $P_i[c_l]$ also must be isomorphic to $S_3$, therefore $P_i[t]$ is a prism that is split (3-3)-wise.
\end{proof}

\begin{remark} \label{obs:SubgraphIndSet}
Let $(T,\delta)$ be a  \HCtree{}, and $t$ a node with children $c_l,c_r$. We give the children of $c_l$ and $c_r$ names $l_1,l_2$ and $r_1,r_2$ respectively. Furthermore, we give the children of these 4 nodes names $x_1,x_2$, $x_3,x_4$, $y_1,y_2$ and $y_3,y_4$ respectively. If $T[c_l]$ and $T[c_r]$ are fully normalized, then for every $i\in\{1,\ldots,4\}$, $G[x_i]$ and $G[y_i]$ have no edges.
\end{remark}

\begin{proof}
Assume that $T[c_l]$ and $T[c_r]$ are fully normalized. By definition, we know that all the subgraphs in $G[c_l]$ and $G[c_r]$ have been split optimally as balanced as possible. This means that all the subgraphs in $G[l_1]$, $G[l_2]$, $G[r_1]$ and $G[r_2]$ have at most 3 vertices. These subgraphs are also split optimally and balanced. This means that for any $T[x_i]$ or $T[y_i]$, every subgraph is isomorphic to either of $\emptyset, S_1, s_2$ and thus have no edges.
\end{proof}

When explaining the operation, we assume that the nodes have the same names as in Remark \ref{obs:SubgraphIndSet}. From here, we identify the nodes that are children of $l_1$, $l_2$, $r_1$ and $r_2$. We then switch around all the subgraphs that are split right-heavy, so they become left-heavy split. Figure \ref{fig:leftheavy} shows this operation. Specifically, we modify $(T,\delta)$ into $(T',\delta')$ such that for each pair of nodes $x_i,y_i\in T'$, we have
$$G_{(T',\delta')}[x_i] = (G_{(T,\delta)}[x_i]\cap G[t]^L)\cup(G_{(T',\delta')}[y_i]\cap G[t]^R)$$
$$G_{(T',\delta')}[y_i] = (G_{(T,\delta)}[x_i]\cap G[t]^R)\cup(G_{(T',\delta')}[y_i]\cap G[t]^L)$$

\begin{figure}[H]
\centering
\includegraphics[width=\textwidth]{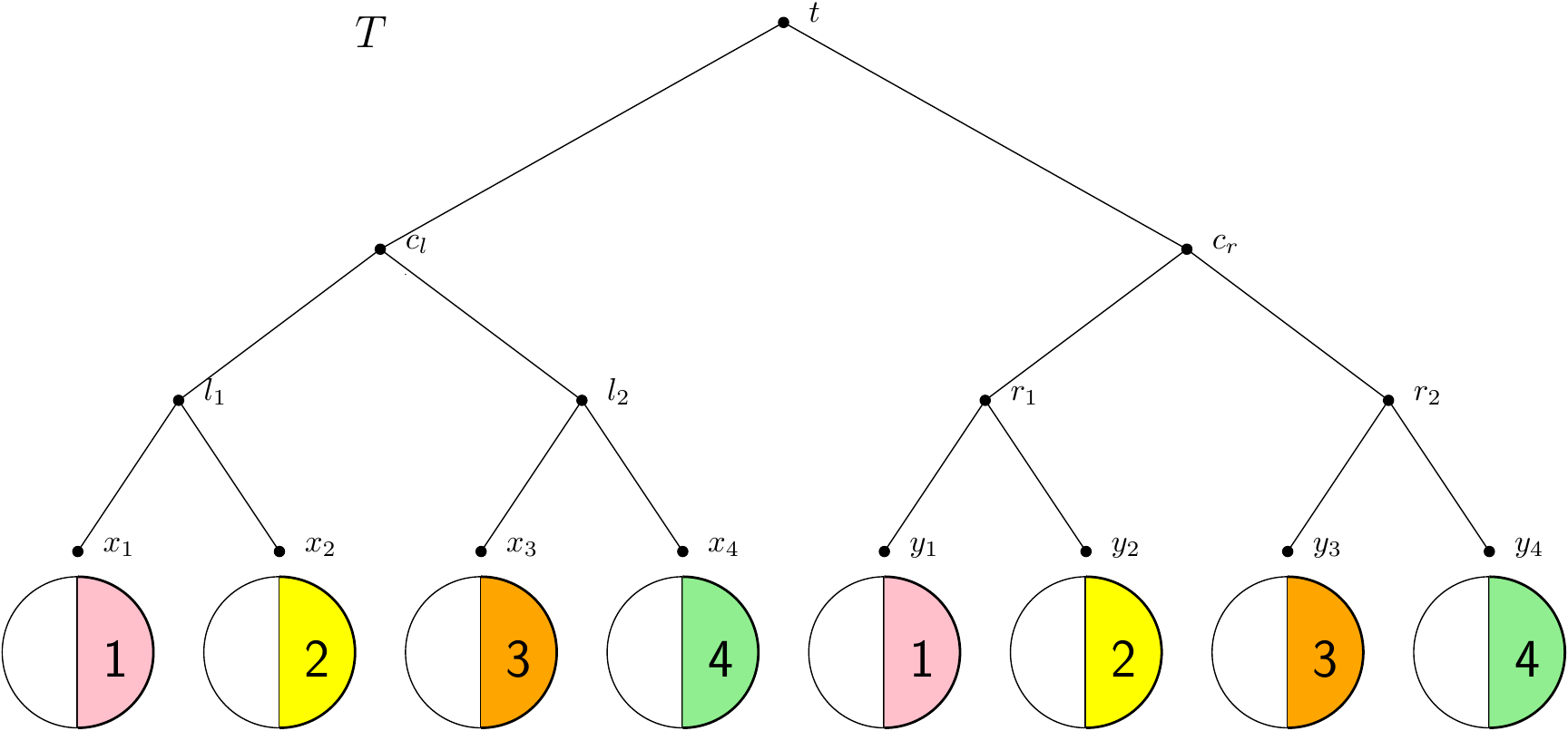}
\caption{The circles beneath each node $x_i$ (or $y_i$) represents $G_{(T,\delta)}[x_i]$ (or $G_{(T,\delta)}[y_i]$); the colored halves represent the sub-prisms that are right-heavily split at $t$, i.e.  the union of all those $P_i[t]$ for which $|P_i[c_l]|<|P_i[c_r]|$. In the Left-Heavy Distribution operation, we switch each two colored parts with the same number.}
\label{fig:leftheavy}
\end{figure}

\begin{lemma}\label{Lemma:safepart2}
Left-Heavy Distribution on any node $t$ is a safe operation.
\end{lemma}

\begin{proof}

As implied by Remark \ref{obs:SubgraphIndSet}, none of the subgraphs $G[x_i]$ or $G[y_i]$ have any edges. This means that for every $i$, any \HCtree{} of $G_{(T',\delta')}[x_i]$ or $G_{(T',\delta')}[y_i]$ has DC-cost zero. When this step is done, every edge in $G[t]$ is cut at one of the nodes $t$, $c_l$, $c_r$, $l_1$ or $l_2$. It is also evident that every edge is cut in a subgraph that is at least as big in $T'$ as it was in $T$, except the edges in $c_r$. Following Remark \ref{Rmk:RightSideEdges}, these edges must necessarily follow from a $S_6\rightarrow(S_3,S_3)$ split at $t$. The decrease in cost for these edges are therefore matched by the increase in cost for the other $S_3$ that is split at $c_l$. It follows that $(T',\delta')$ has at least as high DC-cost as $(T,\delta)$.
Note that every subgraph in $T'[c_l]$ and $T'[c_r]$ is still fully normalized, since they are split the same way as before.
\end{proof}

\subsection{Balancing the \HCtree{}}\label{Sct:Balance}

Let $t$ be a node of \HCtree{} $(T,\delta)$ on which we have just performed Left-Heavy Distribution. 
This means that every split at a node $t$ is optimal and left-heavy, and also that we have performed Balancing on both its children $c_l,c_r$, so that $T[c_l],T[c_r]$ are both fully normalized. 
In the Balancing step we fully normalize $T[t]$.
Since splits at the children are left-heavy, there are 12 possible splits of sub-prisms at $t$ before we perform Balancing.  These are the 8 in Figure \ref{fig:optsplit} plus 4 not cutting any edge. 4 of these 12 (the first 4 in below) are as even as possible, while 8 are uneven. 

\begin{minipage}[ht!]{0.5\textwidth}
\begin{itemize}
\item $a$ splits of type $S_6\rightarrow(S_3,S_3)$
\item $b$ splits of type $S_5\rightarrow(S_3,S_2)$
\item $c$ splits of type $S_4\rightarrow(S_2,S_2)$
\item $d$ splits of type $S_3\rightarrow(S_2,S_1)$
\item $a'$ splits of type $S_6\rightarrow(S_6,\emptyset)$
\item $b'$ splits of type $S_6\rightarrow(S_5,S_1)$
\end{itemize}
\end{minipage}
\begin{minipage}[ht!]{0.5\textwidth}
\begin{itemize}
\item $c'$ splits of type $S_6\rightarrow(S_4,S_2)$
\item $d'$ splits of type $S_5\rightarrow(S_5,\emptyset)$
\item $e'$ splits of type $S_5\rightarrow(S_4,S_1)$
\item $f'$ splits of type $S_4\rightarrow(S_4,\emptyset)$
\item $g'$ splits of type $S_4\rightarrow(S_3,S_1)$
\item $h'$ splits of type $S_3\rightarrow(S_3,\emptyset)$\\
\end{itemize}
\end{minipage}

The Balancing step is done as follows: Each uneven split of a sub-prism is modified into the unique even split on the same sub-prism, by way of moving some vertices from the left side over to the right side. Figure \ref{fig:mod} shows the details of this operation. In the resulting HC-tree, the sub-prisms are not necessarily split left-heavily in $c_l$ or $c_r$ anymore. This does not affect the cost, as these nodes are the lowest that cut edges. We still flip the left and right side of these sub-prisms to guarantee the behavior of performing Left-Heavy distribution on the parent of $t$.

As an example of this type of modification, consider a sub-prism that is split $S_5\rightarrow(S_4,S_1)$ before the modification. We will modify it into $S_5\rightarrow(S_3,S_2)$. In this case, we move one single vertex from the left side to the right side. To optimize the split, we must pick the one vertex that is not adjacent to the vertex already lying on the right side. However, note that these movements of vertices from left subtree to right subtree affect also the cost of edges belonging to even splits, and thus Figure \ref{fig:mod} shows also the effects on even splits.

For every possible split, we have denoted the number of sub-prisms that are split this way at $t$ with a letter as shown above, where the letters $a$ to $d$ are reserved for even splits and ticked letters $a'$ through $h'$ are reserved for uneven splits.

From Remark \ref{obs:SubgraphIndSet}, we know that before the Balancing step at $t$, every edge in $G[t]$ is cut at one of the nodes $t$, $c_l$, $c_r$, $l_1$ and $l_2$ (where the nodes are named as in Figure \ref{fig:leftheavy}). After the modification, every edge in $G[t]$ is cut at one of the nodes $t$, $c_l$ and $c_r$ in $(T', \delta')$. How much is gained and lost for each type of split is shown in Figure \ref{fig:mod}.

\begin{figure}[ht!]
\centering
\includegraphics[width=\textwidth]{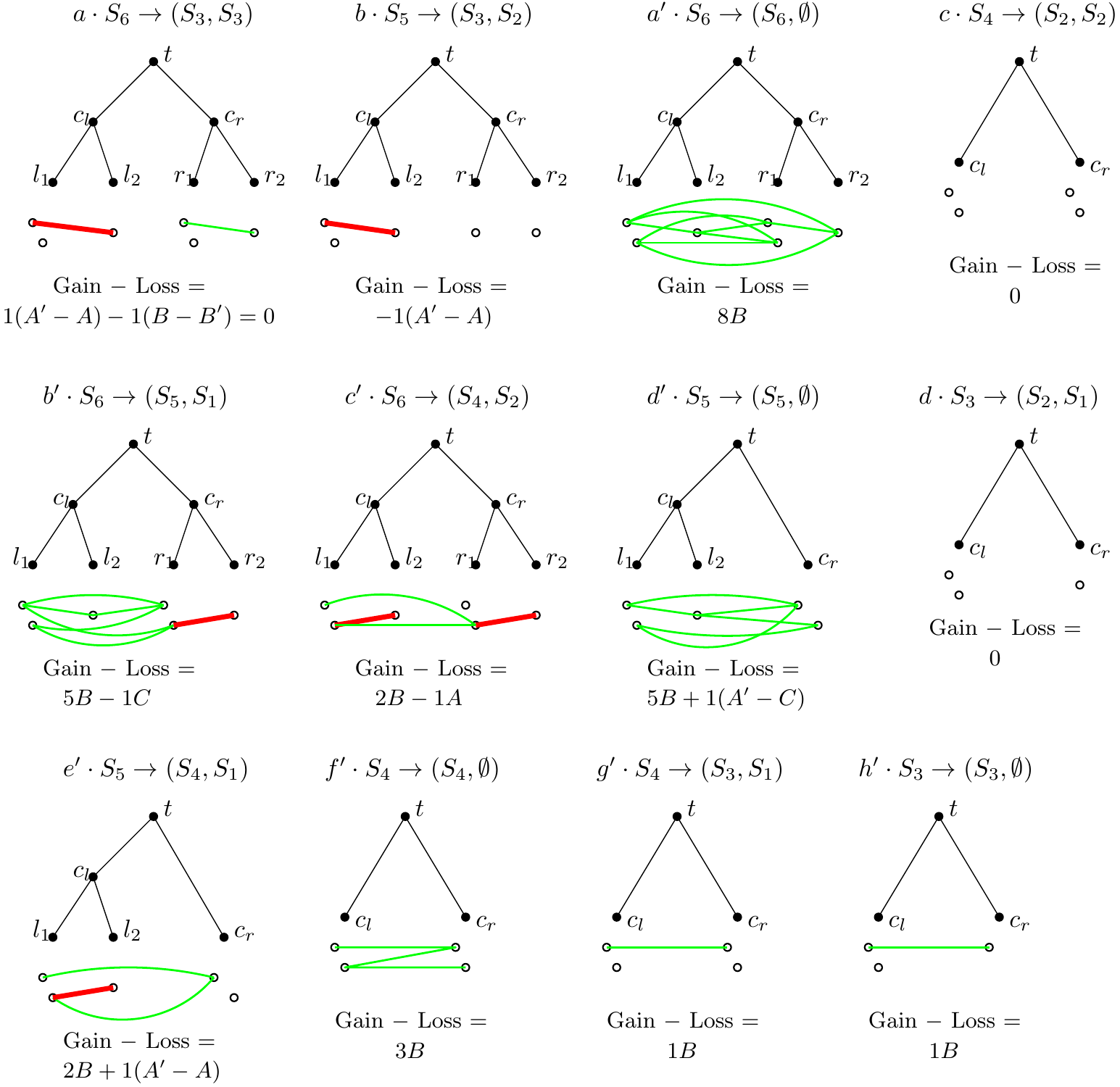}
\caption{This figure shows every type of split that gets some edges modified in the Balancing step, \emph{after} the modification. Green edges have gained cost and red edges have lost cost. Edges whose cost do not change are not shown.}
\label{fig:mod}
\end{figure}

\begin{lemma}\label{Lemma:safepart3}
In the bottom-up traversal the Balancing operations collectively contribute to making this bottom-up traversal a safe operation.
\end{lemma}

\begin{proof}
Assume Balancing has been performed at a node $t$ as explained above, with the letters $a,...,d,a',...h'$ denoting the number of sub-prisms before the Balancing of each of the 12 types. 
To calculate the change in cost, we must look at the sizes of subgraphs of $G[t]$, with $A$ the number of leaves of the subtree rooted at left child before Balancing at $t$ and $A'$ this number after the balancing at $t$, and similarly for $B,B',C$ (remember that $(T,\delta)$ is the tree before this step and $(T',\delta')$ is the modified \HCtree{}):

\begin{itemize}
\item $A := |G_{(T,\delta)}[c_l]| = 6(a')+5(b'+d')+4(c'+e'+f')+3(a+b+g'+h')+2(c+d)$
\item $A' := |G_{(T',\delta')}[c_l]| = 3(a+b+a'+b'+c'+d'+e')+2(c+d+f'+g'+h')$
\item $B := |G_{(T,\delta)}[c_r]| = 3(a)+2(b+c+c')+1(d+e+b'+e'+g')$
\item $B' := |G_{(T',\delta')}[c_r]| = 3(a+a'+b'+c')+2(b+c+d'+e'+f'+g')+1(d+h')$
\item $C := |G_{(T,\delta)}[l_1]| \leq 3(a'+b'+d')+2(a+b+c'+e'+f'+g'+h'+c+d)$
\item $N := |G[t]| = A+B = A'+B'$
\end{itemize}

Back to our example, we see in Figure \ref{fig:mod} that in each of the $e'$ sub-prisms that used to be split $S_5 \rightarrow (S_4, S_1)$ there are 3 edges that have their cost changed, for two of them a gain of $B=(A+B)-A$ since these edges used to be on the left side but are now cut at $t$, while one edge incurs a loss of $A-A'$ since the left side has shrunk in size. The net gain (Gain minus Loss) for these $e'$ sub-prisms is thus $e'(2B-A+A')$.

The net gain for all sub-prisms split at $t$ is found by summing in a similar way the net gain for all the 12 cases. Into this total net gain we now plug the definitions of $A,A',B,B',C,N$ given above, to get a large sum of products of pairs of the variables $a,...,d,a',...,h'$. After a simple, but tedious reorganizing of this sum each pair will be multiplied by a coefficient in this total net gain; these coefficients are shown in Table \ref{tab}.

In this sum, every coefficient is non-negative, except for two terms: $-b'h'$ and $-c'h'$. This means that if $G[t]$ consists of only $S_6\rightarrow(S_4,S_2)$'s (denoted by $c'$) and $S_3\rightarrow(S_3,\emptyset)$'s (denoted by $h'$), then the modified $(T',\delta')$ actually has \emph{lower} DC-cost than the original $(T,\delta)$. In other words, not every call to Balancing will be safe. But in every ancestor of $t$, the $c'$ $S_6\rightarrow(S_4,S_2)$'s are $S_6\rightarrow(S_6,\emptyset)$'s, and the $h'$ $S_3\rightarrow(S_3,\emptyset)$'s will at some ancestor be involved in one of $S_4\rightarrow(S_3,S_1)$, $S_5\rightarrow(S_3,S_2)$ or $S_6\rightarrow(S_3,S_3)$. The coefficients for these combinations in the sum are 8, 13 and 24, respectively. Therefore, even when including these combinations of sub-prisms, the cost for these sub-prisms must increase more at the ancestors of $t$ than it decreases at $t$. The same argument can be put forward for the combination $-b'h'$. This implies that no pair of sub-prisms contributes a lower DC-cost in the finished, factorized HC-tree than at the start of the bottom-up traversal.
\end{proof}

\begin{table}[ht!]
\centering
\begin{tabular}{|c||r|r|r|r|r|r|r|r|}
    \hline
     & $a'$ & $b'$ & $c'$ & $d'$ & $e'$ & $f'$ & $g'$ & $h'$ \\
     \hline \hline
    $a$ & 24 & 13 & 3 & 16 & 6 & 9 & 3 & 3 \\
    \hline
    $b$ & 13 & 6 & 0 & 9 & 3 & 4 & 1 & 1 \\
    \hline
    $c$ & 16 & 8 & 2 & 10 & 4 & 6 & 2 & 2 \\
    \hline
    $d$ & 8 & 3 & 0 & 5 & 2 & 3 & 1 & 1 \\
    \hline
    $a'$ & 0 & 5 & 10 & 0 & 5 & 0 & 8 & 0 \\
    \hline
    $b'$ & x & 2 & 5 & 2 & 3 & 1 & 4 & \textcolor{red}{-1} \\
    \hline
    $c'$ & x & x & 0 & 6 & 1 & 2 & 1 & \textcolor{red}{-1} \\
    \hline
    $d'$ & x & x & x & 0 & 4 & 0 & 5 & 0 \\
    \hline
    $e'$ & x & x & x & x & 1 & 1 & 2 & 0 \\
    \hline
    $f'$ & x & x & x & x & x & 0 & 3 & 0 \\
    \hline
    $g'$ & x & x & x & x & x & x & 1 & 1 \\
    \hline
    $h'$ & x & x & x & x & x & x & x & 0 \\
    \hline
\end{tabular}
\caption{The coefficients associated with each pair of variables, in the formula for net gain after modification of the \HCtree{} $(T[t],\delta)$. That is, net gain is equal to $24aa'+13ab'+\ldots+1g'h'+0h'h'$. Note the two negative numbers.}
\label{tab}
\end{table}

\begin{lemma}\label{Lemma:allsafe}
The top-down traversal of $(T,\delta)$ in which Cut Optimization is performed is a safe operation. The bottom-up traversal of $(T,\delta)$ in which Left-Heavy Distribution and Balancing is performed is a safe operation.
\end{lemma}

\begin{proof}
Lemma \ref{Lemma:safepart1} has already established that the top-down traversal consists of a series of safe operations and is therefore itself a safe operation, i.e. the DC-cost of the HC-tree that was given as input is no higher than the DC-cost of the HC-tree after top-down traversal. 
By Lemma \ref{Lemma:safepart2} the Left-heavy Distribution on each node is also safe. By Lemma \ref{Lemma:safepart3} the combined result of all the Balancing operations together imply that the bottom-up traversal is also a safe operation, i.e. the DC-cost of the HC-tree resulting from the top-down traversal does not have DC-cost higher than the DC-cost of the HC-tree after the bottom-up traversal. 
\end{proof}


\begin{customlemma}{5}
The prism $P$ is max-well-behaved, and thus $C_6$ is min-well-behaved.
\end{customlemma}

\begin{proof}
We have demonstrated a safe normalization procedure that works for any $k$ and any HC-tree of $G=P^{(k)}$ as described by Property \ref{Fact:ExistsProc}. Safeness of the procedure follows from the safeness of the two steps, both the top-down traversal and the bottom-up traversal, as established by Lemma \ref{Lemma:allsafe}. This means that no HC-tree of $G=P^{(k)}$ has DC-cost higher than the tree output by the normalization procedure. This output tree is a factorized HC-tree{} since at its root node $r$ 
every connected subgraph $P_i[r]$ of $G[r]$ is the prism $S_6$ and every prism at $r$ is split into two $S_3$'s, which are further split into the independent sets $S_2$ and $S_1$, as in Figure \ref{fig:prismvars}. This decomposition is thus the factorized \HCtree{}, of DC-cost $48k^2$.
\end{proof}

\section{Conclusion}

We leave as an open problem the complexity of deciding if a graph is max or min well-behaved. A related question arises if we assume that we are given an  HC-tree $T$ of max DC-cost for a graph $H$ and also an integer $k$, and we ask for an HC-tree of max DC-cost for $H^{(k)}$. Note that the equivalent min DC-cost version of this problem, where adjacency denotes similarity, instead looks at the join of $k$ copies, i.e. a dense graph where an edge is added between any two vertices from distinct copies. It is not clear to us if these problems on $k$ copies are solvable in polynomial time, even though we assume an optimal HC-tree is given for a single copy.

\bibliography{Hierarchical-Clustering}

\begin{thebibliography}{10}

\bibitem{BBC04}
N.~Bansal, A.~Blum, and S.~Chawla.
\newblock Correlation clustering.
\newblock {\em Machine Learning}, 56(1-3):89--113, 2004.

\bibitem{Bun71}
P.~Buneman.
\newblock The recovery of trees from measures of dissimilarity.
\newblock {\em Mathematics in the Archaeological and Historical Sciences},
  pages 387--395, 1971.

\bibitem{CEF+06}
S.~Chakrabarti, M.~Ester, U.~Fayyad, J.~Gehrke, J.~Han, S.~Morishita,
  G.~Piatetsky-Shapiro, and W.~Wang.
\newblock Data mining curriculum: A proposal (version 1.0).
\newblock Technical report, Intensive Working Group of ACM SIGKDD, 2006.

\bibitem{CC17}
M.~Charikar and V.~Chatziafratis.
\newblock Approximate hierarchical clustering via sparsest cut and spreading
  metrics.
\newblock In {\em Annual ACM-SIAM symposium on Discrete algorithms (SODA)},
  pages 841--854, 2017.

\bibitem{CKM+19}
V.~Cohen-Addad, V.~Kanade, F.~Mallmann-Trenn, and C.~Mathieu.
\newblock Hierarchical clustering: Objective functions and algorithms.
\newblock {\em Journal of ACM}, 66(4):26:1--26--42, 2019.

\bibitem{Das19}
S.~Dasgupta.
\newblock Hardness of hierarchical clustering optimization.
\newblock Private communication, 2019.

\bibitem{Das16}
S.~Dasgutpa.
\newblock A cost function for similarity-based hierarchical clustering.
\newblock In {\em Annual ACM symposium on Theory of Computing (STOC)}, pages
  118--127, 2016.

\bibitem{Die05}
R.~Diestel.
\newblock {\em Graph theory}.
\newblock Springer-Verlag, 2005.

\bibitem{Har75}
J.~Hartigan.
\newblock {\em Clustering algorithms}.
\newblock John Wiley and Sons, 1975.

\bibitem{HTF09}
T.~Hastie, R.~Tibshirani, and J.~Friedman.
\newblock {\em The elements of statistical learning: data mining, inference,
  and prediction}.
\newblock Springer Series in Statistics. Springer, second edition, 2009.

\bibitem{KT09}
K.~Koutroumbas and S.~Theodoridis.
\newblock {\em Pattern recognition}.
\newblock Academic Press, fourth edition, 2009.

\bibitem{SS63}
R.~Sokal and P.~Sneath.
\newblock {\em Numerical taxonomy}.
\newblock W.H. Freeman, 1963.

\end{thebibliography}
\bibliographystyle{plainurl}

\end{document}